\documentclass[12pt]{article}

\usepackage[margin=2.5cm]{geometry}

\usepackage{amsmath,amsthm,amsfonts,amscd,amssymb,bbm,mathrsfs,enumerate,url}
\usepackage{authblk}

\usepackage{graphicx,tikz}
\usepackage[pdfborder={0 0 0}]{hyperref}
\usetikzlibrary{matrix,arrows}
\usetikzlibrary{decorations.pathreplacing}
\usetikzlibrary{decorations.pathmorphing}
\usetikzlibrary{patterns,fadings}

\def\semicolon{;}
\def\applytolist#1{
    \expandafter\def\csname multi#1\endcsname##1{
        \def\multiack{##1}\ifx\multiack\semicolon
            \def\next{\relax}
        \else
            \csname #1\endcsname{##1}
            \def\next{\csname multi#1\endcsname}
        \fi
        \next}
    \csname multi#1\endcsname}

\def\calc#1{\expandafter\def\csname c#1\endcsname{{\mathcal #1}}}
\applytolist{calc}QWERTYUIOPLKJHGFDSAZXCVBNM;
\def\bbc#1{\expandafter\def\csname bb#1\endcsname{{\mathbb #1}}}
\applytolist{bbc}QWERTYUIOPLKJHGFDSAZXCVBNM;
\def\bfc#1{\expandafter\def\csname bf#1\endcsname{{\mathbf #1}}}
\applytolist{bfc}QWERTYUIOPLKJHGFDSAZXCVBNM;
\def\sfc#1{\expandafter\def\csname s#1\endcsname{{\sf #1}}}
\applytolist{sfc}QWERTYUIOPLKJHGFDSAZXCVBNM;
\def\rfc#1{\expandafter\def\csname r#1\endcsname{{\mathrm #1}}}
\applytolist{rfc}QWERTYUIOPLKJHGFDSAZXCVBNM;
\def\scfc#1{\expandafter\def\csname sc#1\endcsname{{\mathscr #1}}}
\applytolist{scfc}QWERTYUIOPLKJHGFDSAZXCVBNM;

\DeclareMathOperator{\Ad}{Ad}

\DeclareMathOperator{\supp}{supp}
\DeclareMathOperator{\diff}{Diff_+}

\usepackage{xspace}

\def\bb1{\mathbbm{1}}

\def\<{\langle}
\def\>{\rangle}
\def\Mob{\mathrm{M\ddot{o}b}}
\def\conf{{\mathscr C}}

\def\dS{\mathrm{dS}}
\def\lorentz{{\cL^\uparrow_+}}
\def\fin{\text{fin}}

\newtheorem{theorem}{Theorem}[section]

\newtheorem{proposition}[theorem]{Proposition}

\theoremstyle{remark}

\title{Towards integrable perturbation of 2d CFT \\
on de Sitter space}
\author[1]{Christian D.\ J{\"a}kel\thanks{{\tt jaekel@ime.usp.br}}}
\author[2]{Yoh Tanimoto\thanks{{\tt hoyt@mat.uniroma2.it}}}

\affil[1]{Department of Applied Mathematics, University of S\~ao Paulo (USP), \authorcr
   Rua de Mat{\~a}o 1010, CEP 05508-090 S{\~a}o Paulo, Brazil}
\affil[2]{Dipartimento di Matematica, Universit\`a di Roma Tor Vergata,\authorcr
   Via della Ricerca Scientifica 1, I-00133 Roma, Italy}
\date{}

\begin{document}
\maketitle

\begin{abstract}
We describe a procedure to deform the dynamics of a two-dimensional 
conformal net to possibly obtain a Haag-Kastler net on the de Sitter spacetime. 
The new dynamics is given by adding a primary field smeared on the time-zero circle
to the Lorentz generators of the conformal net.
 
As an example, we take an extension of the chiral $\rU(1)$-current net by a charged field
with conformal dimension $d < \frac14$. We show that the perturbing operators are
defined on a dense domain.
\end{abstract}

\section{Introduction}
The first interacting quantum field theories, the $\scP(\phi)_2$-models, have been constructed by starting with the free field
on the Minkowski space,
defining an interaction term, perturbing the dynamics by it locally, finding the interacting vacuum
and changing the Hilbert space \cite{GJ72}. 
The $\scP(\phi)_2$-models have been constructed also
on the de Sitter space \cite{FHN75}, then recently formulated into the operator-algebraic framework \cite{BJM23}.
Perturbing the dynamics on the de Sitter space has the advantage that one may construct
interacting models on the same Hilbert space, as one can avoid Haag's theorem \cite{Weiner11}.

The procedure 
of perturbing the dynamics 
has 
been formulated 
in \cite{JM18}:
One starts with a Haag-Kastler net on the de Sitter space (in the sense of \cite{BB99}),
then alters the Lorentz boosts by defining the new ones as the modular groups 
for a rotation invariant, interacting vacuum vector.
The above $\scP(\phi)_2$-models fit in this programme.
As the arguments do not depend on the properties of free fields,
one may wish to find other examples.
We propose such an example in this work, where the starting QFT is a two-dimensional conformal field theory
and the perturbation is given by a primary field.
Specifically, the conformal field theory is a two-dimensional extension
of the chiral $\rU(1)$-current algebra, and we take the charge-carrying field as the interaction term.
Such fields have been constructed recently as two-dimensional conformal Wightman fields \cite{AGT23Pointed}.
Such a conformal field can be seen as a field on the de Sitter space through a conformal map.

The $\rU(1)$-current algebra is defined on the Hilbert space $\cH_0$, and the two-dimensional
extension contains two copies of it as the left and right chiral components.
The chiral components have charged sectors $\cH_\alpha$ parametrized by $\alpha \in \bbR$.
For a fixed $\alpha_0 \in \bbR$, we take $\bigoplus_{j\in \bbZ} \cH_{j\alpha_0} \otimes \cH_{j\alpha_0}$ as 
the Hilbert space\footnote{From a general point of view, it is more natural to
take $\bigoplus_{j\in\bbZ}\cH_{j\alpha_0} \otimes \cH_{-j\alpha_0}$ (cf.\! \cite{LR95}).
In this case, the resulting net is unitarily equivalent through the map $J(z) \to -J(z)$,
which can be unitarily implemented.
That is, by denoting $\theta$ the adjoint action by this unitary,
one has $\rho_\alpha\circ \theta = \theta \circ \rho_{-\alpha}$.
This should be distinguished from \textbf{unitary (non)equivalence} of sectors, $\Ad U \circ \rho_\alpha \neq \rho_{-\alpha}$
for any $U$ if $\alpha \neq 0$.
Note that $J(z)$ is defined in Section \ref{sec:3.1}.}. 
For $\alpha \in \alpha_0\bbZ$, there is a charged field $Y_\alpha(z)$ that maps $\cH_\beta\otimes \cH_\beta$
to $\cH_{\beta+\alpha}\otimes \cH_{\beta+\alpha}$, $\beta \in \alpha_0\bbZ$. This is the basis of our perturbing field.
We show that the symmetric field $Y_\alpha(z)\otimes Y_\alpha(z^{-1}) + Y_\alpha(z)^* \otimes Y_\alpha(z^{-1})^*$
can be added to the Lorentz generators of the de Sitter space, and they still satisfy the Lorentz relations weakly.

From a physical point of view, perturbing the Hamiltonian of a CFT by a (relevant) field has been proposed
to obtain massive integrable models in \cite{Zamolodchikov89-1}. Depending on the initial CFT and the perturbing field,
various integrable models should be obtained.
While our results are specific to the $\rU(1)$-current, the proof of (weak) Lorentz relations
depends essentially on the fact that we take a primary field that is commutative at the time-zero circle.
Therefore, the idea should generalize to many CFTs and primary fields.

This paper is organized as follows.
In Section 2, we briefly recall the algebraic framework on the de Sitter spacetime,
how a two-dimensional CFT can be considered on the de Sitter spacetime and
the perturbation of the dynamics by a local field.
In Section 3, a family of two-dimensional extensions of the $\rU(1)$-current net and their
charged fields are reviewed. In Section 4, we make estimates of the charged fields
restricted to the time-zero circle and show that the restriction defines operators
if the charge $\alpha$ satisfies $|\alpha| < \frac1{\sqrt 2}$.
Section 5 shows that the time-zero charged fields commute with each other.
In Section 6, we show that the Lorentz generators perturbed by the charged field
still satisfy the Lorentz relations weakly on a certain domain.
In Section 7, we describe how this programme can be completed.

\section{General strategy}\label{general}
\subsection{Haag-Kastler nets on the de Sitter space}
The two-dimensional de Sitter space $\dS^2$ is embedded 
in the ambient three-dimensional Minkowski space $\bbR^{1+2}$
by the equation $x_0^2 - x_1^2 - x_2^2 = -  r^2$, 
where $r>0$.
The isometry group of $\dS^2$ is the (proper orthochronous) Lorentz 
group $\lorentz$ (the connected component of the stabilizer subgroup 
of the point $(0,0,0)$ in the three-dimensional Poincar\'e group), also called the de Sitter group.
On this space, the causal structure and the metric can be introduced 
by restricting those of the ambient Minkowski space.
The region $\{(x_0,x_1,x_2)\in \bbR^{1+2} \mid  | x_0 | < x_1  \}$ is called the wedge in 
the $x_1$-direction.
We denote its intersection with $\dS^2$ by $W_1$. 
Any image of $W_1$ by a Lorentz transformation is called
a \textbf{wedge} in $\dS^2$. For any wedge $W$ there is 
a one-parameter group $\Lambda_W(t)$ of Lorentz boosts
that fix $W$, which are referred to as the boosts associated 
with $W$.

The Haag-Kastler axioms, usually considered on the Minkowski space, 
can be also formulated on $\dS^2$ \cite{BB99},
where the spectrum condition is replaced by the geodesic KMS property 
as below. A \textbf{Haag-Kastler net on $\dS^2$} is a triple $(\cA, U, \Omega)$, 
where $\cA$ is a family of von Neumann algebras on a Hilbert space $\cH$
parametrized by open regions $O \subset \dS^2$,
$U$ is a unitary representation of $\lorentz$ on $\cH$
(continuous in the strong-operator topology) and $\Omega$ is a vector in $\cH$, such that 
\begin{enumerate}[{(HK}1{)}]
\item\label{HK:isotony} {\bf Isotony:} $\cA(O_1) \subset \cA(O_2)$ for $O_1 \subset O_2$;
\item {\bf Locality:} If $O_1$ and $O_2$ are spacelike separated,
then $\cA(O_1)\subset \cA(O_2)'$;
\item {\bf Lorentz covariance:} $\cA(gO) = \Ad U(g)(\cA(O))$ for 
$g \in \lorentz$;
\item {\bf Cyclicity}: $\Omega$ is cyclic for each $\cA(O)$;
\item\label{HK:KMS} {\bf The geodesic KMS property:} 
For any wedge $W$, it holds that $U(\Lambda_W(2\pi t)) = \Delta_W^{-it}$,
where $\Delta_W^{it}$ is the modular group of the algebra $\cA(W)$ with respect to $\Omega$.
\end{enumerate}
The geodesic KMS property is equivalently stated by saying
that $U(\Lambda_W(t))$ satisfies the KMS condition for $\cA(W)$ with 
temperature $2\pi$ 
with respect to $\Omega$ \cite{BB99}.

\subsection{Two-dimensional conformal net on the de Sitter space}
Our starting point is a two-dimensional conformal field theory.
A two-dimensional conformal field theory on the Minkowski space is a theory that is (locally) covariant with respect
not only to the Poincar\'e group but also to the universal covering $\overline{\Mob}\times\overline{\Mob}$ of the M\"obius group
$\Mob \times \Mob$ including special conformal transformations, and often further to $\overline{\diff(S^1)}\times \overline{\diff(S^1)}$.
It is known that any two-dimensional conformal (Haag-Kastler) net, \textit{a priori} defined on the Minkowski space
$\bbR^{1+1}$, extends to the Einstein cylinder \cite{KL04-2}\cite[Theorem A.5]{MT18}.
Then, the de Sitter space $\dS^2$ can be embedded conformally in the Einstein cylinder,
and we can restrict the given conformal net to this subset \cite{GL03}, see Figure \ref{fig:conformal}.
Therefore, a two-dimensional 
CFT can be considered as a QFT on $\dS^2$ in this natural sense.
Let us briefly review how this is done.

In our framework, a conformal Haag-Kastler net can be described as follows.
First, we consider $\bbR^{1+1}$ as the product of two lightrays.
Each lightray $\bbR$ has $S^1$ as the one-point compactification, and the 
group $\diff(S^1)$ acts on it.
Therefore, $\overline{\diff(S^1)}\times \overline{\diff(S^1)}$ acts on $\bbR^{1+1}$ 
locally in the sense of \cite{BGL93}.
Furthermore, by spacelike locality, this action factors through the subgroup
$\mathfrak R := \{R_{2n\pi}\times R_{-2n\pi}:n \in \bbZ\}$ where 
$R_{t} \in \overline{\diff(S^1)}$ is 
the lift of the rotation by $t$ \cite[Proposition~2.1]{KL04-2} 
(see also \cite[Theorem A.5]{MT18}).
We denote this group by $\conf$.
The Minkowski space $\bbR^{1+1}$ is conformally equivalent to the 
product $I_{2\pi}\times I_{2\pi}$
of open intervals of length $2\pi$, and through the local action of $\conf$,
the Haag-Kastler net can be extended to $\bbR\times \bbR$ 
quotiented by the action of $\mathfrak R$,
where $\bbR$ is the universal covering of~$S^1$. 
This space is conformally equivalent to the Einstein 
cylinder $\mathcal{E} = S^1 \times \bbR$ (the product structure is different from the previous one).
We say that two regions $O_1, O_2$ are spacelike separated if there is a diamond
obtained by shifting the Minkowski space $I_{2\pi} \times I_{2\pi}$, which includes $O_1, O_2$
such that $O_1$ and $O_2$ are spacelike separated there.

To be precise, the axioms 
for conformal nets on $\mathcal{E}$ are
the following:
Let $\cA$ be a family of von Neumann algebras on $\cH$ parametrized by open regions in $\mathcal{E}$,
let $U$ be a unitary projective representation of the group $\conf$ on $\cH$
(note that the restriction of $U$ to the subgroup $\overline{\Mob}\times\overline{\Mob}/\mathfrak R$ can be actually made into a true representation
of $\overline{\Mob}\times\overline{\Mob}$ \cite[Theorem~7.1]{Bargmann54},
and the generators of one-parameter subgroups in $\overline{\Mob}\times\overline{\Mob}$ are uniquely defined) and
let $\Omega$ be a vector in $\cH$ 
such that
\begin{enumerate}[{(CN}1{)}]
\item {\bf Isotony:} $\cA(O_1) \subset \cA(O_2)$ for $O_1 \subset O_2$;
\item {\bf Locality:} If $O_1$ and $O_2$ are spacelike separated,
then $\cA(O_1)\subset \cA(O_2)'$;
\item {\bf Conformal covariance:} $\cA(\gamma O) = \Ad U(\gamma)(\cA(O))$ for 
$\gamma \in \conf$, and if $O$ is disjoint from $\supp \gamma$, then $\Ad U(\gamma)(x) = x$
for $x \in \cA(O)$;
\item {\bf Positive energy}: The generators of the chiral rotation subgroups $R_t\times \iota, \iota \times R_t$,
where $\iota \in \overline{\diff(S^1)}$ is the unit element, are positive;
\item {\bf Vacuum}: $\Omega$ is cyclic for each $\cA(O)$ and is a unique (up to scalar) vector such that $U(\gamma)\Omega = \Omega$
for $\gamma \in \overline{\Mob}\times \overline{\Mob}$ in the sense above.
\end{enumerate}
In a conformal net, the Bisognano-Wichmann property holds 
automatically \cite{BGL93}.

\goodbreak
On the Einstein cylinder $\cE$, the
strip of temporal width $\pi$, see Figure \ref{fig:conformal}, 
is conformally equivalent to the de Sitter space \cite{GL03}
(for each radius $r$ there are different de Sitter 
spaces, 
but we do not specify 
$r$,
because 
the only point is that any of such de Sitter 
spaces
is conformally equivalent to the same part of the cylinder).
The time-zero circle $S^1$ ($x_0=0$) in the de Sitter space is 
the (compactified) time-zero line ($a_0=0$) on the cylinder,
and space rotations act on it.
Other Lorentz transformations are contained in the conformal 
group $\conf$. Indeed, the spacelike rotations $R_t\times R_{-t}$ 
and Lorentz boosts (that is the product of lightlike dilations with opposite sign) generate
the three-dimensional Lie group $\mathrm{SO}(2,1)$ (the $(2+1)$-dimensional Lorentz group),
also called the $2d$ de Sitter group.
Therefore, by restricting a conformal net to the de Sitter space,
it satisfies the axioms (HK\ref{HK:isotony}-\ref{HK:KMS}): the 
geodesic KMS property is satisfied because
of the Bisognano-Wichmann property.

\begin{figure}[ht]\centering
\begin{tikzpicture}[scale=0.75]
        \begin{scope}
        \draw [->] (-3,0) --(5,0) node [above right] {$a_1$};
         \draw [->] (0,-3)--(0,4) node [ right] {$a_0$};
          \draw [thick] (-2,0)-- (0,2) node [ left] {$M_0$};
          \draw [thick] (0,2)-- (2,0);
          \draw [thick] (-2,0)-- (0,-2);
          \draw [thick] (2,0)-- (0,-2);
          
           \draw [ thick] (0,0)-- (1,1);
          \draw [thick] (0,0)-- (1,-1);
          
\draw [ thick] (-0.4,0) -- (-0.9,0.5);
\draw [thick] (-0.4,0) -- (-0.9,-0.5);
\draw [thick] (-1.4,0) -- (-0.9,0.5);
\draw [thick] (-1.4,0) -- (-0.9,-0.5);
          
          \draw [thick,dotted] (-2,-3)--(-2,4);
          \draw [thick,dotted] (2,-3)--(2,4);
          
          \draw [thin,->] (-3,-3) -- (3.5,3.5) node [above right] {$a_R$};
          \draw [thin,->] (3,-3) -- (-3.5,3.5) node [above left] {$a_L$};
  \fill [color=black,opacity=0.2]
               (-0.4,0) -- (-0.9,0.5) -- (-1.4,0) --(-0.9,-0.5);
           \fill [color=black,opacity=0.6]
              (0,0) -- (1,1) -- (2,0) -- (1,-1);
        \end{scope}

        \begin{scope}[shift={(10,0)}]
        \draw [->] (-3,0) --(5,0) node [above right] {$a_1$};
         \draw [->] (0,-3)--(0,4) node [ right] {$a_0$};
          \draw [thick] (2,1)-- (-2,1) -- (-2,-1)--(2,-1)--cycle;
          \node at (-0.5,1.5) {$\mathrm{dS^2}$};
           \draw [ thick] (0,0)-- (1,1);
          \draw [thick] (0,0)-- (1,-1);
          
\draw [ thick] (-0.4,0) -- (-0.9,0.5);
\draw [thick] (-0.4,0) -- (-0.9,-0.5);
\draw [thick] (-1.4,0) -- (-0.9,0.5);
\draw [thick] (-1.4,0) -- (-0.9,-0.5);
          
          \draw [thick,dotted] (-2,-3)--(-2,4);
          \draw [thick,dotted] (2,-3)--(2,4);
          
          \draw [thin,->] (-3,-3) -- (3.5,3.5) node [above right] {$a_R$};
          \draw [thin,->] (3,-3) -- (-3.5,3.5) node [above left] {$a_L$};
  \fill [color=black,opacity=0.2]
               (-0.4,0) -- (-0.9,0.5) -- (-1.4,0) --(-0.9,-0.5);
           \fill [color=black,opacity=0.6]
              (0,0) -- (1,1) -- (2,0) -- (1,-1);
        \end{scope}
\end{tikzpicture}
\caption{The Minkowski space $M_0$ (cf.\! \cite[Figure 1]{AGT23Pointed}) and the de Sitter space $\mathrm{dS}^2$
conformally embedded in $\bbR^2$. The cylinder is obtained by identifying the dotted lines.
The dark grey region is a wedge $W$ and the light grey region is a double cone.}
\label{fig:conformal}
\end{figure}
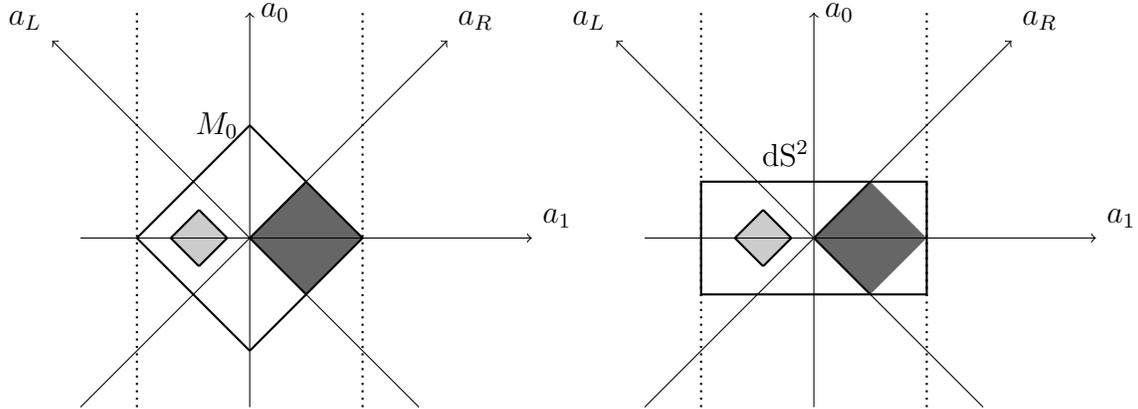

In general, a two-dimensional conformal net $\cA$ contains \textbf{chiral components},
that are observables living on the lightrays and invariant under the action of
$\iota \times \overline{\diff(S^1)}$ ($\overline{\diff(S^1)} \times \iota$, respectively).
They can be regarded as Haag-Kastler nets on the lightrays $\bbR$.  They
extend to $S^1$ by conformal covariance and are called \textbf{conformal nets on $S^1$}. 
More precisely, a triple 
$(\cA_0, U_0, \Omega_0)$, where $\cA_0$ is a family of von Neumann algebras on 
a Hilbert space~$\cH_0$ parametrized by
open connected nonempty non dense intervals in $S^1$, $U_0$ is a unitary projective representation
of $\diff(S^1)$ and $\Omega_0$ is a vector in $\cH_0$,
is called a \textbf{conformal net on $S^1$} if it satisfies
\begin{enumerate}[{(CNS}1{)}]
\item {\bf Isotony:} $\cA_0(I_1) \subset \cA_0(I_2)$ for $I_1 \subset I_2$;
\item {\bf Locality:} If $I_1$ and $I_2$ are disjoint,
then $\cA_0(I_1)\subset \cA_0(I_2)'$;
\item {\bf Conformal covariance:} $\cA_0(\gamma I) = \Ad U_0(\gamma)(\cA_0(I))$ for 
$\gamma \in \diff(S^1)$, and if $I$ is disjoint from $\supp \gamma$, then $\Ad U_0(\gamma)(x) = x$
for $x \in \cA_0(I)$;
\item {\bf Positive energy}: The 
generator 
of the rotation subgroup 
in $\diff(S^1)$ is positive;
\item {\bf Vacuum}: $\Omega_0$ is cyclic for each 
$\cA_0(I)$ and is a unique (up to a scalar)
vector such that $U_0(\gamma)\Omega_0 = \Omega_0$ for $\gamma \in \Mob$;
\end{enumerate}
A two-dimensional conformal net $\cA$ contains both left and right chiral components. 
Indeed, the operators $U(\gamma_0\times \iota), U(\iota\times \gamma_0)$ 
are such elements.
In addition, left and right chiral components commute 
with each other \cite{Rehren00}.

The positive-energy representation $U_0$ is associated with a positive-energy representation
of the Virasoro algebra $\{L_m\}$ \cite[Appendix]{Carpi04}:
\[
 [L_m, L_n] = (m-n)L_{m+n} + \frac{c}{12}m(m^2-1)\delta_{m,-n},
\]
for a certain value $c > 0$ (by an abuse of notations, we use the symbols $\{L_m\}$ both for
abstract Lie algebra elements and for unbounded operators 
satisfying 
the above relations).
The self-adjoint operators $\mathfrak t = \tfrac{1}{2}L_0 - \frac14(L_1 + L_{-1}), 
\mathfrak d = \tfrac{i}{2}(L_{-1} - L_1)$ are the generators
of translations and dilations of 
$\bbR \subset S^1$
(in the sense that $\bbR$ is embedded in $S^1$ by the stereographic projection), 
respectively \cite[Appendix A]{WeinerThesis}.

The two-dimensional conformal group $\conf$ has the tensor product $\{L_m \otimes \bb1, \bb1\otimes L_n\}$
as the Lie algebra. The Lorentz boosts are generated by 
$\mathfrak{k}_1 := \mathfrak d\otimes \bb1 - \bb1\otimes \mathfrak d$,
while the 
rotations of the time-zero circle of the Einstein cylinder (identified with that of the de Sitter space)
are generated by $\mathfrak{k}_0 := L_0\otimes \bb1 - \bb1\otimes L_0$.
Note that the second boost is given by $\mathfrak{k}_2 := i[ \mathfrak{k}_1 ,  \mathfrak{k}_0 ]$. 

For $z \in S^1$, one can consider the operator-valued distribution $T(z) = \sum_n L_n z^{-n}$,
called the Virasoro field (the convention of the exponent is the one such that $T(z)^* = T(z)$ \cite{CTW19},
different from $L(z) = \sum_n L_n z^{-n-2}$ in vertex algebras \cite{Kac98}).
The two-dimensional stress-energy tensor $T^{\mu\nu}$ has four components,
and the fields $T(z)\otimes \bb1 + \bb1\otimes T(z^{-1})$ and $T(z)\otimes \bb1 - \bb1\otimes T(z^{-1})$
correspond to the components of the stress-energy tensor $T^{01}$ and $T^{00}$ \cite{RehrenCQFT}.

\subsection{Perturbation by a local field}\label{perturbation}
As we saw above, a two-dimensional conformal net can be considered as a Haag-Kastler net on $\dS^2$ and hence 
as a starting point for a
new construction
in the sense of \cite{JM18}. We take a starting Haag-Kastler net $\cA$ on $\dS^2$, a unitary representation $U$ of the $(2+1)$-dimensional Lorentz group
and a vacuum vector $\Omega$.
We call ``time-zero wedges'' wedge regions whose end points reside on $S^1$.

The general strategy of \cite[Theorem 4.1]{JM18} goes as follows:
\begin{itemize}
\item Assume there is a rotation-invariant vector $\Omega$ cyclic for $\cA(W_1)$ in 
the natural positive cone $\cP( \cA(W_1), \Omega_0)$ associated with the pair $\cA (W_1) $ and the free 
vacuum vector $\Omega_0$. The rotation-invariance of $\Omega$ implies that 
$\Omega \in \cP( \cA(W), \Omega_0)$ for all time-zero wedges $W$.

\item Assume there exists a new 
(interacting) 
representation $\widetilde U$ of the Lorentz 
group such that 
\begin{itemize}
\item its restriction to the rotation subgroup coincides with that of $U$; 
\item its (interacting) Lorentz boost associated with the wedge $W_1$
is implemented by the modular group for the pair $\cA(W_1)$ 
and (the interacting vacuum vector) $\Omega$;
\item it satisfies the finite speed of light 
condition \cite[Definition 3.3]{JM18}, which roughly says that
the action of $\widetilde U$ on the 
time-zero algebras $\cA(O_I)$, with $O_I$ a double cone given by the intersection 
of time-zero wedges, preserves locality.
\end{itemize}

The last two assumptions should be satisfied automatically if the new 
representation $\widetilde U$ is generated from a local field, as below\footnote{For this implication,
Trotter's product formula was used in \cite[Theorem 10.1.1]{BJM23}. For this, it is necessary
that the new generators are essentially self-adjoint on the same domain. Otherwise, new techniques would be needed.}.
\item The Lorentz covariance of the new net is given by $\widetilde U$. For any 
wedge region $W$, $\widetilde \cA(W)
:= \Ad \widetilde U(g)(\cA(W_1))$ 
(by definition), 
where $g$ is such that $gW_1 = W$
(this is well-defined by finite speed of propagation, in particular,
the free time-zero wedge algebras are preserved by the new boosts). 
Any double cone is written 
as $O = \bigcap_{W \supset O} W$, and accordingly
we define $\widetilde \cA(O) = \bigcap_{W \supset O} \widetilde \cA(W)$. 
This satisfies locality again by finite speed of propagation.
\end{itemize}

Then $(\widetilde \cA, \widetilde U, \widetilde \Omega)$ is a new Haag-Kastler net on $\dS^2$.
This construction avoids Haag's theorem \cite{Weiner11}, because the spacetime is compact and there is no
dilation covariance that pushes a double cone to infinity.

Assume that the net $\cA$ is generated by a conformal Wightman field $\psi$
and it has a well-defined restriction to 
the $x_0 = 0$ circle, which we denote by $\psi(0,\theta)$.
The field $\psi(0,\theta)$ smeared by a test function $f$ on the circle is
denoted by $\psi(0,f)$.
Let $\mathfrak e_n(\theta) = e^{in\theta}$. 
We find interesting candidates for such $\widetilde U$ by adding 
$\psi(0,\mathfrak e_1), \psi(0,\mathfrak e_{-1})$
to 
the building blocks 
$\mathfrak{l}_1 := L_1 \otimes \bb1 + \bb1 \otimes L_{-1}, 
\mathfrak{l}_{-1} :=L_{-1}\otimes \bb1 + \bb1 \otimes L_1$
of the generators 
\begin{align*}
	\mathfrak{k}_1 = \tfrac{1}{2} (\mathfrak{l}_1 + \mathfrak{l}_{-1}) , 
	\qquad 
	\mathfrak{k}_2 = \tfrac{1}{2i} (\mathfrak{l}_1 - \mathfrak{l}_{-1}) , 
\end{align*}
of the Lorentz boosts and $\psi(0,\mathfrak e_0)$ 
to the generator of the rotations 
$L_0 \otimes \bb1 - \bb1 \otimes L_0$ leaving 
the $x_0 = 0$ circle invariant. By definition, 
$\mathfrak{l}_1 = \mathfrak{k}_1 + i \mathfrak{k}_2$  and 
$\mathfrak{l}_{-1} = \mathfrak{k}_1 - i \mathfrak{k}_2 $, 
and consequently, $\mathfrak{l}_1^* = \mathfrak{l}_{-1}$.

\bigskip
Below we take concrete examples of two-dimensional CFTs
and a candidate for $\widetilde U$ using the charged fields in it.

\section{The \texorpdfstring{$\rU(1)$}{U(1)}-current and its two-dimensional extension}
\label{u1current}

\subsection{Chiral components}
\label{sec:3.1}
To make the programme concrete, we consider a
two-dimensional CFT,
whose chiral components are the $\rU(1)$-current nets
on $S^1$ \cite{BMT88}.
It is generated by the current (the derivative of the massless free field) $J(z) = \sum_n J_n z^{-n-1}$,
where $z \in S^1$, and its Fourier coefficients $J_n$ satisfy the commutation relations
\begin{align}\label{eq:current}
 [J_m, J_n] = n\delta_{m,-n}.
\end{align}
There is a representation of this algebra with a unique \textbf{vacuum vector} $\Omega_0$ such that
$J_n \Omega_0 = 0$ for $n \ge 0$. 
The Hilbert space $\cH_0$ is spanned by the vectors of the form
\[
 J_{-n_1}\cdots J_{-n_k}\Omega_0,
\]
where $n_1 \ge n_2 \ge \cdots \ge n_k$.
We denote by $\cH_0^{\fin}$ the linear span of these vectors.
$\cH_0$ is equipped with an inner product
$\<\cdot, \cdot\>$ with respect to which it holds that $J_n^* = J_{-n}$.

This current can be smeared by a smooth function $f$ and gives an unbounded operator $J(f) = \sum_n f_n J_n$,
where $f_n = \frac1{2\pi}\int e^{-in\theta} f(\theta)d\theta$ are the Fourier components.
The exponential $W(f) = e^{iJ(f)}$ is called a Weyl operator.
One can construct a conformal net on $S^1$ by $\cA_0(I) = \{e^{iJ(f)}: \supp f \subset I\}''$.
A representation of the Virasoro algebra $\{L_n\}$ is given by the Sugawara formula
$L_n = \frac12 \sum_k :J_{n-k} J_k:$, and it integrates to a projective unitary representation $U_0$ of $\diff(S^1)$.
With this $U_0$, $(\cA_0, U_0, \Omega_0)$ is called the \textbf{$\rU(1)$-current net}.

A \textbf{representation} of a conformal net $\cA_0$ on $S^1$ is a family of isomorphisms $\{\rho_I\}$
of local algebras $\{\cA_0(I)\}$ to von Neumann
algebras on a certain Hilbert space $\cH_\rho$
that satisfy the compatibility condition
\[
 \rho_{I_1}(a) 
= \rho_{I_2}(a) \qquad \text{for } 
a \in \cA(I_1), \quad I_1 \subset I_2.
\]

For each $\alpha \in \bbR$, the $\rU(1)$-current net admits a representation $\rho_\alpha$.
This representation $\rho_\alpha$ can be realized on the same Hilbert space $\cH_0$ as the vacuum representation,
but to distinguish them we denote it by $\cH_\alpha$ and the lowest weight vector by $\Omega_\alpha$.
The assignment $W(f) \mapsto \rho_\alpha(W(f)) = e^{i\alpha \int f(\theta)d\theta}W(f)$
gives the representation.
This is an \textbf{automorphism} of each local algebra $\cA_0(I)$.
In terms of generators, it amounts to 
replacing $J_0$ (which acts as the 
zero operator 
in the vacuum representation)
by the scalar $\alpha$. We denote the generator of the current algebra \eqref{eq:current} on this space by $J_{\alpha,n}$.
By the Sugawara formula 
$L_{\alpha,n} = \frac12 \sum_k :J_{\alpha, n-k} J_{\alpha,k}:$,
there is also a representation of the Virasoro algebra with the same central charge $c=1$.
We denote by $\cH_\alpha^{\fin}$ the 
subspace spanned by vectors of the form
$J_{\alpha,-n_1}\cdots J_{\alpha,-n_k}\Omega_\alpha$.

Two such automorphisms $\rho_{\alpha_1}, \rho_{\alpha_2}$ can be composed,
and yield a new automorphism (representation) $\rho_{\alpha_1 + \alpha_2}$.
This composition law is called the fusion rule for the $\rU(1)$-current.
For a fixed $\alpha_0 \in \bbR$, the family $\{\rho_{j\alpha_0}\}_{j\in \bbZ}$ 
on $\{\cH_{j\alpha_0}\}$ is closed under fusion, indeed, 
$\rho_{j_1\alpha_0}\rho_{j_2\alpha_0} = \rho_{(j_1+j_2)\alpha_0}$.

\subsection{Charged primary fields}
Let us fix $\alpha_0 \in \bbR$ with $|\alpha_0| \le 1, \alpha_0 \neq 0$. We now
construct\footnote{The case $\alpha=0$ is possible but uninteresting because $Y_0(z) = \bb1$, thus we do not consider it although not explicitly excluded.}
for $\alpha \in \alpha_0\bbZ$, on a dense domain in the Hilbert space
$\hat \cH = \bigoplus_{j \in \bbZ} \cH_{j\alpha_0}$,
a formal series\footnote{Recall \cite{AGT23Pointed} that we use formal series 
$\sum_{s \in \bbR} A_s z^s$ for a family of operators $\{A_s\}_{s \in \bbR}$
(actually the formal series is just the parametrized family $\{A_s\}$ itself, 
but certain operations on them
are implicit).} by
	\begin{align}
  		Y_\alpha(z) &= \sum_{s \in \bbR} Y_{\alpha, s}z^{-s-d},
	\end{align}
where $d = \frac{\alpha^2}2$, as follows.
On $\hat \cH$, the operators 
	\[
		\hat J_n = \bigoplus_{j\in \bbZ} J_{j\alpha_0, n} \, , 
		\qquad \hat L_n = \bigoplus_{j\in \bbZ} L_{j\alpha_0, n} 
	\]
can be defined naturally (depending on $\alpha_0$, they act on a Hilbert space $\hat \cH = \bigoplus_{j \in \bbZ} \cH_{j\alpha_0}$
which also depends implicitly on $\alpha_0$).
Let $c_\alpha$ be the unitary charge shift
operator $\cH_\beta \to \cH_{\beta+\alpha}$ defined
by $c_\alpha \hat J_{-n_1}\cdots \hat J_{-n_k}\Omega_{\beta} = \hat J_{-n_1}\cdots \hat J_{-n_k}\Omega_{\beta + \alpha}$, $n_j > 0$.
We can regard $c_\alpha$ as an operator on $\hat \cH = \bigoplus_{j \in \bbZ} \cH_{j\alpha_0}$.
Following \cite{TZ12}, we define
	\begin{align}\label{eq:E}
 		E^\pm(\alpha, z) 
		&= \exp \Bigl(\mp \sum_{n>0} \frac{\alpha \hat J_{\pm n}}n z^{\mp n}\Bigr) . 
	\end{align}
The operators $Y_\alpha(z)$ are now specified by
	\begin{align}
  		Y_\alpha(z) &= c_\alpha E^-(\alpha,z) E^+(\alpha,z) z^{\alpha J_0} ,
		\label{eq:Eaz} 
	\end{align}
where $z^{\alpha J_0}$ means $z^{\alpha \beta}$ on $\cH_\beta$.
Each coefficient $Y_{\alpha, s}$ is the direct sum of maps $\cH_\beta \to \cH_{\beta+\alpha}$
and on each $\cH_\beta, \beta \in \alpha_0\bbZ$, only $Y_{\alpha,s}$ with $s \in \bbZ - \alpha\beta - \frac{\alpha^2}2 = \bbZ - \alpha\beta - d$ are non-zero.

The operators 
$\hat L_n$
generate a projective unitary representation $\hat U_0$ of $\overline{\diff(S^1)}$.
We know from \cite{AGT23Pointed} (see \cite{TZ12} and \cite{Toledano-LaredoThesis} for the original references) that,
if $|\alpha| \le 1$, then $Y_{\alpha}$ satisfies,
with $d = \frac{\alpha^2}2$ as above,
\begin{align}
 [\hat L_m, Y_{\alpha,s}] &= ((d-1)m - s)Y_{\alpha, m+s}, \label{eq:primary} \\
 \|Y_{\alpha,s}\| &\le 1. \label{eq:energybounds}
\end{align}
By \eqref{eq:energybounds}, we can smear the field by a test function supported in $\bbR$
and obtain a bounded operator $Y_{\alpha}(f)$
(even for $|\alpha| > 1$ one can define $Y_{\alpha}(f)$, but they are unbounded).
By \eqref{eq:primary}, it is conformally covariant with respect to $\hat U$ with the conformal dimension $d$.
The formal series $Y_\alpha(z) = \sum_{s \in \bbR} Y_{\alpha, s} z^{-s-d}$ is not local with itself, but satisfies a braiding relation.

From the construction \eqref{eq:Eaz},
it is easy to see the commutation relation
$[\hat J_m, Y_{\alpha}(z)] = \alpha Y_\alpha(z) z^{m}$, or equivalently,
\begin{align}
 [\hat J_m, Y_{\alpha, s}] &= \alpha Y_{\alpha, m+s}. \label{eq:commjy}
\end{align}
This implies that $Y_\alpha(z)$ is relatively local to $J(w)$.

\subsection{Two-dimensional Wightman field}\label{twodim}
From two copies of a conformal nets $\cA_0$ on $S^1$, one can construct a two-dimensional conformal net
by tensor product: $\cA_0(I_+\times I_-) := \cA_0(I_+)\otimes \cA_0(I_-)$.
The unitary representation of $\conf$ is given by $U_0(\gamma_+) \otimes U_0(\gamma_-)$
and the vacuum by $\Omega_0\otimes \Omega_0$.
For $\alpha_0 \in \bbR, |\alpha_0|$ fixed above,
there is an extension of this net on the space $\bigoplus_{j\in\bbZ} \cH_{j\alpha_0} \otimes \cH_{j\alpha_0}$ \cite{MTW18, AGT23Pointed}:
on each direct summand $\cH_{j\alpha_0} \otimes \cH_{j\alpha_0}$,
the tensor product net $\cA_0$ acts by the representation $\rho_{j\alpha_0}\otimes \rho_{j\alpha_0}$.

Furthermore, for $\alpha \in \alpha_0\bbZ$, let $Y_{\alpha, s}$ as above on $\hat \cH = \bigoplus_{j \in \bbZ} \cH_{j\alpha_0}$.
The components $Y_{\alpha, s}\otimes \bb1$ of the charged field $Y_{\alpha}(z)\otimes \bb1$ of the left chiral component
acts on $\hat \cH \otimes \hat \cH$
trivially on the right chiral component.
Similarly, the components $\bb1\otimes Y_{\alpha, s}$ of the charged field $\bb1\otimes Y_{\alpha}(z)$ of the right chiral component
acts on $\hat \cH \otimes \hat \cH$,
trivially on the left chiral component.

For $\alpha \in \alpha_0 \bbZ$, we consider the combined charged field
\[
 \widetilde\psi^\alpha(w,z) = Y_\alpha(w)\otimes Y_\alpha(z) + (Y_\alpha(w)\otimes Y_\alpha(z))^*,
\]
which is a formal series whose coefficients are operators acting on $\hat \cH \otimes \hat \cH$.
Actually, as the components of $Y_\alpha(w)\otimes Y_\alpha(z)$ raise
(respectively the components of $(Y_\alpha(w)\otimes Y_\alpha(z))^*$ lower)
the left and right charges by $\alpha$ at the same time,
they restrict to $\widetilde \cH = \bigoplus_{j\in\bbZ} \cH_{j\alpha_0}\otimes \cH_{j\alpha_0}$.
As we take $|\alpha| \le 1$, by \cite[Theorem 5.9]{AGT23Pointed},
the field $\widetilde\psi^\alpha(w,z)$ is a two-dimensional conformal Wightman field
that generates a two-dimensional Haag-Kastler net. Let us call this net $\widetilde \cA$.

\section{Estimates for the charged fields}\label{estimates}
As we wish to perturb the conformal net by a charged field, we are interested in the time-zero restriction of
$\widetilde\psi^\alpha(w,z) = Y_\alpha(w)\otimes Y_\alpha(z) + (Y_\alpha(w)\otimes Y_\alpha(z))^*$.
This amounts to taking
$z = w^{-1}$. However, it is \textit{a priori} unclear whether this is possible, because
taking $z= w^{-1}$ should give a formal series of $z$ alone, but each component is an infinite sum
of components of $\widetilde\psi^\alpha(w,z)$:
\begin{align*}
 Y_\alpha(w)\otimes Y_\alpha(z) = \sum_{s \in \bbR} Y_{\alpha, s}w^{-s-d} \otimes \sum_{t \in \bbR} Y_{\alpha, t}z^{-t-d}
\end{align*}
Therefore, by substituting $z = w^{-1}$,
\begin{align*}
Y_\alpha(w)\otimes Y_\alpha(w^{-1})
&= \sum_{s \in \bbR} Y_{\alpha, s}w^{-s-d} \otimes \sum_{t \in \bbR} Y_{\alpha, t}w^{t+d} \\
 &= \sum_{s \in \bbR} \sum_{t \in \bbR} Y_{\alpha, t} \otimes  Y_{\alpha, t-s}w^{-s} \\
\end{align*}
and we have to make sure that
the sum $\sum_{t\in\bbR} Y_{\alpha, t}\otimes Y_{\alpha, t-s}$ (the sum is countable on each $\cH_{j\alpha_0}\otimes \cH_{j\alpha_0}$)
gives a finite result on a certain dense domain.
As the dense domain, we take
$\bigoplus_{j,\mathrm{alg}}\cH_{j\alpha_0}^{\fin}\otimes_\mathrm{alg} \cH_{j\alpha_0}^{\fin}$
where $\bigoplus_{j,\mathrm{alg}}$ denotes the algebraic direct sum and $\otimes_\mathrm{alg}$
denotes the algebraic tensor product.
We show that the above sum is convergent if $|\alpha| < \frac1{\sqrt 2}$.

For this purpose,
we need a general result on primary fields with conformal weight $d$ \cite[Appendix B (141)]{CKLW18}.
This is proven for fields with integer $d$, but it is straightforward to generalize it
(because one only needs the primarity). It states that
\begin{align}\label{eq:decay}
 \|Y_{\alpha, -n-d}\Omega\|^2 = \binom{2d+n-1}{n} = \frac{\Gamma(2d+n)}{\Gamma(n+1)\Gamma(2d)}\sim n^{2d-1},
\end{align}
where the last asymptotic follows from the Stirling's approximation of the Gamma function $\Gamma(x)$ with complex variable
\cite[12.33]{WW21}.

We first observe that, for $|\alpha| \ge \frac1{\sqrt 2}$, $d = \frac{\alpha^2}2 \ge \frac14$, thus
$\sum_{n\in\bbZ} Y_{\alpha, -n}\otimes Y_{\alpha, -n-s}$ does not converge on $\Omega_0\otimes \Omega_0$.
Indeed,
$\|Y_{\alpha, -n-d}\otimes Y_{\alpha, -n-d-s}\cdot \Omega_0\otimes\Omega_0\|^2 = \binom{2d+n-1}{n}\binom{2d+n+s-1}{n+s} \sim n^{4d-2}$
(for a fixed $s$, as $n \to \infty$) and these vectors are orthogonal to each other,
hence the sum $\sum_{n\in\bbZ} Y_{\alpha, -n-d}\otimes Y_{\alpha, -n-d-s}$ diverges on $\Omega_0\otimes\Omega_0$.
This means that, if $|\alpha| \ge \frac1{\sqrt 2}$, there is no hope 
to define an operator of the form
$\sum_n Y_{\alpha, -n-d}\otimes Y_{\alpha, -n-d-s}$ 
on a domain containing $\Omega_0\otimes \Omega_0$.

On the other hand, if $d < \frac14$, there is still hope that we can 
carry through the general programme of Section \ref{perturbation}.

\begin{theorem}
 Let\footnote{As $\alpha \in \alpha_0 \bbZ$, we can take such an $\alpha$ if $|\alpha_0| < \frac1{\sqrt2}$.} $|\alpha| < \frac1{\sqrt 2}$.
 Then each coefficient of $w^s$ in the formal series $\widetilde\psi^\alpha(w,w^{-1})$,
 applied to any vector in $\bigoplus_{j,\mathrm{alg}}\cH_{j\alpha_0}^{\fin}\otimes_\mathrm{alg} \cH_{j\alpha_0}^{\fin}$,
 is convergent.
\end{theorem}
\begin{proof}
We note first that a general vector in one tensor component is a linear combination of
$c_{j\alpha}\hat J_{-m_1}\cdots \hat J_{-m_k}\Omega_0$, where $m_\ell > 0$ and $c_{j\alpha}$ commutes with  $\hat J_m, m \neq 0$.
Therefore, it is enough to prove the convergence on vectors
in $\cH_{0}^{\fin}\otimes_\mathrm{alg} \cH_{0}^{\fin}$.

We claim that $\|Y_{\alpha, -n-d}\hat J_{-m_1}\cdots \hat J_{-m_k}\Omega_0\|^2 \le C_1(n+C_2)^{4d-2}$
where $C_1, C_2$ depend on the vector but not on $n$.
This is clear for $k=0$.
To prove the claim by induction on $k$, let us observe that $Y_{\alpha, -n-d}\hat J_{-m_1}\cdots \hat J_{-m_k}\Omega_0$ can be reduced, using
the commutation relations \eqref{eq:commjy}, $[Y_{\alpha, -n-d}, \hat J_m] = -\alpha Y_{\alpha, -n+m-d}$,
as follows:
\begin{align*}
 Y_{\alpha, -n-d}\hat J_{-m_1}\cdots \hat J_{-m_k}\Omega_0
 &= ([Y_{\alpha, -n-d}, \hat J_{-m_1}] + \hat J_{-m_1}Y_{\alpha, -n-d})\hat J_{-m_2}\cdots \hat J_{-m_k}\Omega_0 \\
 &= (-\alpha Y_{\alpha, -n+m_1-d} + \hat J_{-m_1}Y_{\alpha, -n-d})
 \hat J_{-m_2}\cdots \hat J_{-m_k}\Omega_0 \, . 
\end{align*}
Using $\|\Psi_1 + \Psi_2\|^2 \le 2(\|\Psi_1\|^2 + \|\Psi_2\|^2)$,
it is enough to show that the norm of each term decays as desired.
The first term decays by the induction hypothesis.
Let us calculate the norm of the second term:
\begin{align*}
 & \|\hat J_{-m_1}Y_{\alpha, -n-d}\hat J_{-m_2}\cdots \hat J_{-m_k}\Omega_0\|^2 \\
 &= \<Y_{\alpha, -n-d}\hat J_{-m_2}\cdots \hat J_{-m_k}\Omega_0, ([\hat J_{m_1}, \hat J_{-m_1}] + \hat J_{-m_1}\hat J_{m_1})Y_{\alpha, -n-d}\hat J_{-m_2}\cdots \hat J_{-m_k}\Omega_0\> \\
 &= \<Y_{\alpha, -n-d}\hat J_{-m_2}\cdots \hat J_{-m_k}\Omega_0, (m_1 + \hat J_{-m_1}\hat J_{m_1})Y_{\alpha, -n-d}\hat J_{-m_2}\cdots \hat J_{-m_k}\Omega_0\> \\
 &= m_1\|Y_{\alpha, -n-d}\hat J_{-m_2}\cdots \hat J_{-m_k}\Omega_0\|^2 + \|\hat J_{m_1}Y_{\alpha, -n-d}\hat J_{-m_2}\cdots \hat J_{-m_k}\Omega_0\|^2 \\
 &= m_1\|Y_{\alpha, -n-d}\hat J_{-m_2}\cdots \hat J_{-m_k}\Omega_0\|^2 + \|(\alpha Y_{\alpha, -n+m_1-d} + Y_{\alpha, -n-d}\hat J_{m_1})\hat J_{-m_2}\cdots \hat J_{-m_k}\Omega_0\|^2 \\
 &\le m_1\|Y_{\alpha, -n-d}\hat J_{-m_2}\cdots \hat J_{-m_k}\Omega_0\|^2 \\
 & \qquad +2 \|\alpha Y_{\alpha, -n+m_1-d}\hat J_{-m_2}\cdots \hat J_{-m_k}\Omega_0\|^2 + 2\|Y_{\alpha, -n-d}\hat J_{m_1}\hat J_{-m_2}\cdots \hat J_{-m_k}\Omega_0\|^2  \, . 
\end{align*}
The last term can be reduced, 
using $[\hat J_m, \hat J_n] = m\delta_{m,-n}$ and $\hat J_m\Omega_0$ (as $m>0$),
to a sum of norms of vectors of the above form. This completes the induction.
That is, the norm of $Y_{\alpha, -n-d}\hat J_{-m_1}\cdots \hat J_{-m_k}\Omega_0$
is a linear combination of terms that decay like $(n+C_2)^{2d-1}$.

When the operator $Y_{\alpha, -n-d}\otimes Y_{\alpha, -n-d-s}$ for a fixed $s$ is
applied to a vector that is the tensor product of two such vectors,
the norm decays as $C_1(n+C_2)^{4d-2}$, which is summable in $n$.
Therefore, this operator is defined on
$\bigoplus_{j,\mathrm{alg}}\cH_{j\alpha_0}^{\fin}\otimes_\mathrm{alg} \cH_{j\alpha_0}^{\fin}$.
\end{proof}

\section{Commutativity of the time-zero charged field}\label{comm}
Now we know that, for $|\alpha| < \frac1{\sqrt 2}$, $Y_{\alpha}(w)\otimes Y_{\alpha}(w^{-1})$ makes sense as a formal series
whose coefficients are (unbounded) operators on the dense domain $\bigoplus_{j,\mathrm{alg}}\cH_{j\alpha_0}^{\fin}\otimes \cH_{j\alpha_0}^{\fin}$.
Next we show that it is not only local but also commutative, and moreover,
$Y_{\alpha}(w)\otimes  Y_{\alpha}(w^{-1})$ and $Y_{\beta}(z)\otimes  Y_{\beta}(z^{-1})$ commute for possibly different $\alpha, \beta \in \alpha_0\bbZ$.
More precisely, we have the following result.
\begin{theorem}\label{th:commutativity}
As formal series, it holds that $(Y_{\alpha}(w)\otimes  Y_{\alpha}(w^{-1})^* = Y_{-\alpha}(w)\otimes  Y_{-\alpha}(w^{-1})$
under the convention $w^* = w^{-1}$ and
$Y_{\alpha}(w)\otimes  Y_{\alpha}(w^{-1})$ and $Y_{\beta}(z)\otimes  Y_{\beta}(z^{-1})$
commute on the domain $\bigoplus_{j,\mathrm{alg}}\cH_{j\alpha_0}^{\fin}\otimes_\mathrm{alg} \cH_{j\alpha_0}^{\fin}$
weakly, that is, as sesquilinear forms.
\end{theorem}
\begin{proof}
The first claim follows easily from the definitions \eqref{eq:E}\eqref{eq:Eaz} and $c_\alpha^* = c_{-\alpha}$.

Let us denote by $\hat J_m\otimes \bb1$ the operator on 
$\bigoplus_{j,\mathrm{alg}}\cH_{j\alpha_0}^{\fin}\otimes_\mathrm{alg} \cH_{j\alpha_0}^{\fin}$
that acts as $\hat J_{j\alpha, m}\otimes \bb1$ on each component $\cH_{j\alpha_0}^{\fin}\otimes_\mathrm{alg} \cH_{j\alpha_0}^{\fin}$.
We calculate the commutator as a sesquilinear form, that is, we apply the operators to a vector
and take a scalar product with another vector, but we omit them.
We have $[\hat J_m, Y_{\alpha, s}] = \alpha Y_{\alpha, m+s}$.
In terms of formal series, this amounts to $[Y_\alpha(w), \hat J_m] = -\alpha Y_\alpha(w)w^{m}$. 
Thus, it holds that, by the 
derivation property of the commutator with the generator, 
\begin{align*}
-&[[Y_\alpha(w)\otimes  Y_\alpha(w^{-1}), Y_\beta(z)\otimes  Y_\beta(z^{-1})], 
\hat J_m \otimes \bb1]\\
&\qquad
= [[\hat J_m\otimes \bb1, Y_\alpha(w)\otimes  Y_\alpha(w^{-1})],  
Y_\beta(z)\otimes  Y_\beta(z^{-1})] \\
&\qquad \qquad +  [[Y_\beta(z)\otimes  Y_\beta(z^{-1}), 
\hat J_m\otimes \bb1], Y_\alpha(w) \otimes  Y_\alpha(w^{-1})]\\
 &\qquad = 
(\alpha w^{m} + \beta z^{m})
[Y_\alpha(w)\otimes  Y_\alpha(w^{-1}),  
 Y_\beta(z)\otimes  Y_\beta(z^{-1})],
\end{align*}
and similarly, with $\bb1\otimes \hat J_m$
that acts as $\bb1\otimes \hat J_m$ on each 
component $\cH_{j\alpha_0}^{\fin}\otimes_\mathrm{alg} \cH_{j\alpha_0}^{\fin}$,
\begin{align*}
&-[[Y_\alpha(w)\otimes  Y_\alpha(w^{-1}), 
Y_\beta(z)\otimes  Y_\beta(z^{-1})], \bb1\otimes \hat J_m] \\
&\qquad = 
(\alpha w^{-m} + \beta z^{-m})
[Y_\alpha(w)\otimes  Y_\alpha(w^{-1}),  
 Y_\alpha(z)\otimes  Y_\alpha(z^{-1})] \, .
\end{align*}
We observe that, upon commuting with
$\hat J_m\otimes\bb1$ or $\bb1\otimes \hat J_m$,
we obtain the same operator 
$[Y_\alpha(w)\otimes  Y_\alpha(w^{-1}),  Y_\alpha(z)\otimes  Y_\alpha(z^{-1})]$
multiplied by a scalar.

Now, to show that the commutator $[Y_\alpha(w)\otimes  Y_\alpha(w^{-1}), Y_\alpha(z)\otimes  Y_\alpha(z^{-1})]$ vanishes,
we only have to check that the matrix
element
vanishes.
It is easy to check that $Y_\alpha(w)\otimes  Y_\alpha(w^{-1})$ commutes with $c_\alpha\otimes c_\alpha$,
therefore, we only have to consider 
pairs 
of vectors in $\cH_0\otimes\cH_0$ and $\cH_\alpha\otimes\cH_\alpha$.
Furthermore, due to the above commutation relations, the 
linear functional
\[
\<\hat J_{-m_1}\cdots \hat J_{-m_k}\Omega_{2\alpha}
\otimes \hat J_{-n_1}\cdots \hat J_{-n_\ell}\Omega_{2\alpha},
\quad\cdot \quad\hat J_{-m'_1}\cdots \hat J_{-m'_{k'}}
\Omega_0\otimes \hat J_{-n'_1}\cdots \hat J_{-n'_{\ell'}}\Omega_0\> 
\]
can be reduced to
the case $\<\Omega_{2\alpha}\otimes\Omega_{2\alpha}, \cdot\;\Omega_0\otimes\Omega_0\>$.

\goodbreak
Let us put
$\underline{Y}_\alpha(z) = E^-(\alpha, z)E^+(\alpha, z)$ and
$\widetilde{\underline{Y}}_\alpha(z) 
= \underline{Y}_\alpha(z)\otimes {\underline{Y}}_\alpha(z^{-1})$.
As we have $E^+(\alpha, z)^* = E^-(-\alpha, z)$ (with the convention that $z^* = z^{-1}$),
it follows that $\underline{Y}_\alpha(z)^* = \underline{Y}_{-\alpha}(z)$
and $\widetilde{\underline{Y}}_\alpha(z)^* = \widetilde{\underline{Y}}_{-\alpha}(z)$,
or equivalently $\widetilde{\underline{Y}}_{m,\alpha}^\dagger 
= \widetilde{\underline{Y}}_{-m,-\alpha}$.

Expand $\widetilde{\underline{Y}}_\alpha(z)$ as
$\widetilde{\underline{Y}}_\alpha(z) 
= \sum_{m \in \mathbb{Z}} \widetilde{\underline{Y}}_{\alpha,m} z^{-m}
= \sum_m \sum_k \underline{Y}_{\alpha,k}\otimes \underline{Y}_{\alpha,k-m} z^{-m}$.
Then we have
\begin{align*}
 Y_{\alpha}(z) \otimes Y_{\alpha}(z^{-1})
 &= (c_\alpha\otimes c_\alpha) z^{\alpha J_0} z^{-\alpha J_0} \underline{Y}_{\alpha}(z)\otimes \underline{Y}_{\alpha}(z^{-1}) \\
 &= (c_\alpha\otimes c_\alpha) \widetilde{\underline{Y}}_{\alpha}(z). 
\end{align*}
Therefore, the question is further reduced to
$\<\Omega_0\otimes \Omega_0, [\widetilde {\underline Y}_{\alpha,m}, \widetilde {\underline Y}_{\beta,n}]\Omega_0\otimes \Omega_0\> = 0$ for all $m,n$.

Note that, with $F$ the flip operator between the left and right tensor components in $\cH_0\otimes \cH_0$
(which is a unitary operator),
\begin{align*}
\widetilde{\underline{Y}}_{\alpha,m}\Omega_0\otimes\Omega_0
&= \sum_k {\underline Y}_{\alpha,k}\otimes  {\underline Y}_{\alpha,k-m}\Omega_0\otimes\Omega_0
= F\cdot \sum_k {\underline Y}_{\alpha,k-m}\otimes  {\underline Y}_{\alpha,k}\Omega_0\otimes\Omega_0 \\
&= F\cdot \sum_k {\underline Y}_{\alpha,k}\otimes  {\underline Y}_{\alpha,k+m}\Omega_0\otimes\Omega_0 
= F\cdot \widetilde {\underline Y}_{\alpha,-m}\Omega_0\otimes\Omega_0 \, .
\end{align*}
This implies that
\begin{align*}
 \<\widetilde{\underline{Y}}_{\alpha,m}\Omega_0\otimes\Omega_0, \widetilde{\underline{Y}}_{\beta,n}\Omega_0\otimes\Omega_0\>
 &= \<F \cdot\widetilde{\underline{Y}}_{\alpha,-m}\Omega_0\otimes\Omega_0,\; F \cdot \widetilde{\underline{Y}}_{\beta,-n}\Omega_0\otimes\Omega_0\> \\
 &= \<\widetilde{\underline{Y}}_{\alpha,-m}\Omega_0\otimes\Omega_0, \; \widetilde{\underline{Y}}_{\beta,-n}\Omega_0\otimes\Omega_0\> \, .
\end{align*}
Moreover, note that the map $J_m \mapsto -J_m$ is a vacuum-preserving automorphism
implemented by a unitary (the multiplication by $(-1)^k$ on the $k$-particle space), and 
its
tensor product maps
\begin{align*}
\widetilde{\underline{Y}}_{\alpha}(z) 
&= E^-(\alpha,z)E^+(\alpha,z)
\otimes  E^-(\alpha,z^{-1}) E^+(\alpha,z^{-1}) \\
&\longmapsto E^-(-\alpha,z)E^+(-\alpha,z)\otimes  E^-(-\alpha,z^{-1}) E^+(-\alpha,z^{-1}) = \widetilde{\underline{Y}}_{-\alpha}(z) \, . 
\end{align*}
That is, $\widetilde{\underline{Y}}_{\alpha, m}$ is mapped 
to $\widetilde{\underline{Y}}_{-\alpha,m}$,
therefore, by the invariance of $\Omega_0\otimes\Omega_0$ by this unitary,
\[
\<\widetilde{\underline{Y}}_{\alpha,m}\Omega_0\otimes\Omega_0, \;\widetilde{\underline{Y}}_{\beta,n}\Omega_0\otimes\Omega_0\>
= \<\widetilde{\underline{Y}}_{-\alpha,m}\Omega_0\otimes\Omega_0, \;\widetilde{\underline{Y}}_{-\beta,n}\Omega_0\otimes\Omega_0\> \, . 
\]
From this, we can compute the commutator (as a sesquilinear form)
\begin{align*}
 &\<\Omega_0\otimes\Omega_0, [\widetilde{\underline{Y}}_{\alpha,m}, \widetilde{\underline{Y}}_{\beta,n}]\Omega_0\otimes\Omega_0\> \\
 &= \<\Omega_0\otimes\Omega_0, (\widetilde{\underline{Y}}_{\alpha,m}\widetilde{\underline{Y}}_{\beta,n} - \widetilde{\underline{Y}}_{\beta,n}\widetilde{\underline{Y}}_{\alpha,m})\Omega_0\otimes\Omega_0\>\\
 &= \<\widetilde{\underline{Y}}_{-\alpha,-m}\Omega_0\otimes\Omega_0, \widetilde{\underline{Y}}_{\beta,n}\Omega_0\otimes\Omega_0\>
  - \<\widetilde{\underline{Y}}_{-\beta,-n}\Omega_0\otimes\Omega_0, \widetilde{\underline{Y}}_{\alpha,m}\Omega_0\otimes\Omega_0\> \\
 &= \<\widetilde{\underline{Y}}_{-\alpha,m}\Omega_0\otimes\Omega_0, \widetilde{\underline{Y}}_{\beta,-n}\Omega_0\otimes\Omega_0\>
  - \<\widetilde{\underline{Y}}_{-\beta,-n}\Omega_0\otimes\Omega_0, \widetilde{\underline{Y}}_{\alpha,m}\Omega_0\otimes\Omega_0\> \\
 &= \<\widetilde{\underline{Y}}_{\alpha,m}\Omega_0\otimes\Omega_0, \widetilde{\underline{Y}}_{-\beta,-n}\Omega_0\otimes\Omega_0\>
  - \<\widetilde{\underline{Y}}_{-\beta,-n}\Omega_0\otimes\Omega_0, \widetilde{\underline{Y}}_{\alpha,m}\Omega_0\otimes\Omega_0\>.
\end{align*}
Furthermore, by construction \eqref{eq:Eaz} of $\widetilde{\underline{Y}}_m = \sum_{k \in \bbR}\underline{Y}_{\alpha, k}\otimes  \underline{Y}_{\alpha, k-m}$,
these expectation values give only real numbers. Therefore, by hermitianity of the scalar product,
this commutator vanishes on the vacuum state.
\end{proof}

\section{Perturbation by charged fields}
According to the general idea of Section \ref{general},
we wish to perturb the net $\widetilde \cA$ by a field by the methods of Barata-J\"akel-Mund \cite{BJM23}.
That is, while keeping the $T^{11}$ component of the stress-energy tensor,
we add a smeared local field to the $T^{00}$ component on the time-zero circle $S^1$.

The necessary condition for it to work is that the new operators satisfy the
Lorentz relations (the generators of the Lorentz group are complex linear combinations of $\mathfrak{l}_m, m=-1,0,1$, cf.\! Section \ref{perturbation}):
\begin{align*}
 [\mathfrak{l}_m, \mathfrak{l}_n] = (m-n)\mathfrak{l}_{m+n}, \quad m,n = -1,0,1.
\end{align*}
As we do not know whether the smeared field can be multiplied on the domain of the old generators,
we consider the weak commutation relation: for two vectors $\Psi_1, \Psi_2$,
we compute $\<A^*\Psi_1, B\Psi_2\> - \<B^*\Psi_1, A\Psi_2\>$.
Obviously, if the commutator $[A,B]$ can be defined on the domain and calculated,
then it implies the weak commutation relation.

We pick the symmetric field
\[
\widetilde\psi^\alpha(w,z) = Y_{\alpha}(w)\otimes Y_{\alpha}(z) + (Y_{\alpha}(w)\otimes Y_{\alpha}(z))^*
=Y_{\alpha}(w)\otimes Y_{\alpha}(z) + Y_{-\alpha}(w)\otimes Y_{-\alpha}(z), 
\]
restrict it to the time-zero circle $z= w^{-1}$ and smear it with
$\mathfrak{e}_j(\theta), n = -1,0,1$. Correspondingly, we consider the coefficients of $z^{-n}$:
\[
\sum_{k \in \bbR} Y_{\alpha,k}\otimes  Y_{\alpha,-n+k} + \sum_{s \in \bbR} Y_{-\alpha,s}\otimes  Y_{-\alpha,-n+s}
= \sum_{k \in \bbR, \epsilon = \pm 1} Y_{\epsilon\alpha,k}\otimes  Y_{\epsilon\alpha,-n+k}.
\]
\begin{theorem}
The Lorentz relations are weakly satisfied for
   \begin{align*}
    & \hat L_1\otimes\bb1 + \bb1 \otimes \hat L_{-1} 
    + \lambda \sum_{k \in \bbR, \epsilon 
    = \pm 1} Y_{\epsilon\alpha,k}\otimes  Y_{\epsilon\alpha,-1+k} \, , 
    \\
    & \hat L_0\otimes\bb1 - \bb1\otimes \hat L_0   \, ,  \\
    & \hat L_{-1}\otimes\bb1 + \bb1\otimes \hat L_1 
    + \lambda \sum_{k \in \bbR, \epsilon 
    = \pm 1} Y_{\epsilon\alpha,k}\otimes  Y_{\epsilon\alpha,1+k} \, , 
   \end{align*}
where $k$ runs in $\bbR$, but there are only countable nonzero terms on 
each $\cH_\beta\otimes\cH_\beta$,
on the domain $\bigoplus_{j,\mathrm{alg}}\cH_{j\alpha_0}^{\fin}
\otimes_\mathrm{alg} \cH_{j\alpha_0}^{\fin}$.
\end{theorem}
\begin{proof}
It is clear that, if $\lambda=0$, these are the old generators and they satisfy
the Lorentz relations on the domain
$\bigoplus_{j,\mathrm{alg}}\cH_{j\alpha_0}^{\fin}
\otimes_\mathrm{alg} \cH_{j\alpha_0}^{\fin}$.
The new terms commute with each other
in the weak sense as we have seen in Section \ref{comm},
hence we only have to check the commutation relations between 
the old terms and the new terms.

One new term can be applied to a vector in the domain and gives a convergent series,
hence we can compute the weak commutator term by term.
For a primary field $Y_\alpha(z)$, it holds that
$[\hat L_m, Y_{\alpha,n}] = ((d-1)m-n)Y_{\alpha,m+n}$ in the operator sense, and hence also in the weak sense.
Therefore, we have the following commutation relation in the weak sense:
\begin{align*}
	& \left[\hat L_m\otimes \bb1 + \bb1\otimes \hat L_{-m} \, , \, 
	\sum_{k \in \bbR} Y_{\alpha,k}\otimes  Y_{\alpha, -n+k} \right] 
	\\
	& \quad =  \sum_{k \in \bbR} \left(((d-1)m-k) Y_{\alpha, k+m} \otimes  Y_{\alpha, -n+k} 
		+ ((d-1)(-m)-(-n+k)) Y_{\alpha, k} \otimes  Y_{\alpha, -m-n+k}\right) 
		\\
	& \quad =  \sum_{k \in \bbR} \left(((d-1)m-(k-m)) Y_{\alpha, k} \otimes  Y_{\alpha, -m-n+k}\right. \\
		&\qquad\qquad \left.+ ((d-1)(-m)-(-n+k)) 	Y_{\alpha, k} \otimes  Y_{\alpha, -m-n+k}\right) \\
	& \quad =  \sum_{k \in \bbR} (m+n-2k) Y_{\alpha, k} \otimes  Y_{\alpha, -m-n+k} \, . 
\end{align*}
As $\alpha$ is arbitrary, this holds even if $\alpha$ is replaced by $-\alpha$.
Furthermore, $\sum_{k \in \bbR} Y_{\alpha,k}\otimes  Y_{\alpha, -n+k}$
and $\sum_{k \in \bbR} Y_{-\alpha,k}\otimes  Y_{-\alpha, -n+k}$
commute by Theorem \ref{th:commutativity}.
Altogether,
\begin{align*}
	& \left[ \hat L_m \otimes \bb1 + \bb1 \otimes \hat L_{-m} + \lambda\sum_{k \in \bbR, \epsilon = \pm 1} Y_{\epsilon\alpha,k}\otimes  Y_{\epsilon\alpha, -m+k},\right. \\
	&\qquad\qquad \left.\hat L_n\otimes \bb1 + \bb1\otimes \hat L_{-n} + \lambda\sum_{k \in \bbR, \epsilon = \pm 1} Y_{\epsilon\alpha,k}\otimes  Y_{\epsilon\alpha, -n+k}
	\right] 
	\\
	& \quad =  (m-n)\hat L_{m+n}\otimes \bb1 - (m-n) \bb1 \otimes \hat L_{-m-n} 
	\\
 	& \quad\;\; + \lambda\sum_{k \in \bbR, \epsilon = \pm 1} (m+n-2k) Y_{\epsilon\alpha, k} \otimes  Y_{\epsilon\alpha, -m-n+k} 
	- \lambda\sum_{k \in \bbR, \epsilon = \pm 1} (n+m-2k) 	Y_{\epsilon\alpha, k} \otimes  Y_{\epsilon\alpha, -n-m+k} \\
	& \quad =  (m-n)\hat L_{m+n}\otimes \bb1 - (m-n) \bb1 \otimes \hat L_{-m-n} \, ,
\end{align*}
and for $m=1, n=-1$, this is $2(\hat L_0\otimes \bb1 - \bb1\otimes \hat L_0)$.

On the other hand,
\begin{align*}
	 & \left[ \hat L_0  \otimes \bb1 - \bb1\otimes \hat L_0 , \; \hat L_m \otimes \bb1 + \bb1 \otimes \hat L_{-m} 
		+ \lambda\sum_{k \in \bbR, \epsilon = \pm 1} Y_{\epsilon\alpha,k}\otimes  Y_{\epsilon\alpha, -m+k} \right] 
		\\
	&\quad =  (-m)\hat L_{m}\otimes \bb1 - (m) \bb1 \otimes \hat L_{-m} \\
	&\qquad + \lambda\sum_{k \in \bbR, \epsilon = \pm 1} \left((-k) Y_{\epsilon\alpha, k} \otimes  Y_{\epsilon\alpha, -m+k} 
	- (-(-m+k)) Y_{\epsilon\alpha, k} \otimes  Y_{\epsilon\alpha, -m+k}\right) 
	\\
	& \quad =  (-m)\left(\hat L_{m}\otimes \bb1 + \bb1\otimes \hat L_{-m} 
	+ \lambda\sum_{k \in \bbR, \epsilon = \pm 1}  Y_{\epsilon\alpha, k} \otimes  Y_{\epsilon\alpha, -m+k}\right) \, .
\end{align*}
For $m=1,-1$, we obtain the right commutation relations 
between $\mathfrak{l}_0$ and $\mathfrak{l}_m$.
\end{proof}

Note that the Lorentz relations do not extend beyond $m=1,0,-1$,
that is, they do not satisfy the Virasoro relations.

In order to implement the perturbation by this commutative field,
we need to solve the following
problems: show 
that the above generators are self-adjoint on a certain domain
and generate a dynamics that satisfies finite speed of propagation.

\bigskip

On the other hand, on the time-zero circle, there is a new representation 
of the Virasoro algebra
(with non-positive energy) with $c=0$, or the Witt algebra.

\begin{proposition}\label{pr:virasoro}
The Virasoro relations are weakly satisfied with $c=0$ for
 \begin{align*}
  & \hat L_m\otimes\bb1 - \bb1\otimes \hat L_{-m} + i\lambda m\sum_{k \in \bbR, \epsilon = \pm 1} Y_{\epsilon\alpha,k}\otimes  Y_{\epsilon\alpha,-m+k},
 \end{align*}
 where $k$ runs in $\bbR$, but there are only countable nonzero terms on each $\cH_\beta\otimes\cH_\beta$.
\end{proposition}
\begin{proof}
 As before, we compute the commutation relations weakly. First,
 with $d = \frac{\alpha^2}2 = \frac{(-\alpha)^2}2$,
\begin{align*}
 &\left[\hat L_m\otimes \bb1 - \bb1\otimes \hat L_{-m}, \sum_{k \in \bbR, \epsilon = \pm 1} Y_{\epsilon\alpha,k}\otimes  Y_{\epsilon\alpha, -n+k}\right] \\
 &=  \sum_{k \in \bbR, \epsilon = \pm 1} \left(((d-1)m-k) Y_{\epsilon\alpha, k+m} \otimes  Y_{\epsilon\alpha, -n+k}\right. \\
 &\qquad\quad - \left.((d-1)(-m)-(-n+k)) Y_{\epsilon\alpha, k} \otimes  Y_{\epsilon\alpha, -m-n+k}\right) \\
 &=  \sum_{k \in \bbR, \epsilon = \pm 1} \left(((d-1)m-(k-m)) Y_{\epsilon\alpha, k} \otimes  Y_{\epsilon\alpha, -m-n+k}\right. \\
 &\qquad\qquad \left.- ((d-1)(-m)-(-n+k)) Y_{\epsilon\alpha, k} \otimes  Y_{\epsilon\alpha, -m-n+k}\right) \\
 &=  \sum_{k \in \bbR, \epsilon = \pm 1} ((2d-1)m -n) Y_{\epsilon\alpha, k} \otimes  Y_{\epsilon\alpha, -m-n+k}
\end{align*}
and by plugging this into the full expressions,
\begin{align*}
	& \left[ \hat L_m\otimes \bb1 - \bb1 \otimes \hat L_{-m} 
	+ i\lambda m\sum_{k \in \bbR, \epsilon = \pm 1} Y_{\epsilon\alpha,k}\otimes  Y_{\epsilon\alpha, -m+k},\right. \\
	&\qquad \qquad \left.\hat L_n\otimes \bb1 - \bb1\otimes \hat L_{-n} + i\lambda n\sum_{k \in \bbR, \epsilon = \pm 1} Y_{\epsilon\alpha,k}\otimes  Y_{\epsilon\alpha, -n+k} \right] 
	\\
	&=  (m-n)\hat L_{m+n}\otimes \bb1 + (-m-n) \bb1 \otimes \hat L_{-m-n} 
	\\
	& \qquad + i\lambda\sum_{k \in \bbR, \epsilon = \pm 1} (n((2d-1)m -n) - m((2d-1)n -m)) 
	Y_{\epsilon\alpha, k} \otimes  Y_{\epsilon\alpha, -m-n+k} 
	\\
	&=  (m-n)\left(\hat L_{m+n}\otimes \bb1 - \bb1\otimes \hat L_{-m-n} 
	+ i\lambda (m+n) \sum_{k \in \bbR, \epsilon = \pm 1} Y_{\epsilon\alpha, k} \otimes  Y_{\epsilon\alpha, -m-n+k}\right) \, .
\end{align*}
\end{proof}
The combination $i\lambda m$ means that we are taking the derivative $-i\partial_\theta\widetilde\psi^\alpha(e^{i\theta},e^{-i\theta})$.
The operators $\hat L_m\otimes \bb1 - \bb1 \otimes \hat L_{-m}$
are the generators of
the time-zero Virasoro (Witt) algebra, and 
Proposition \ref{pr:virasoro}
tells that there are different actions of the Virasoro algebra with $c=0$.

In addition, a formal calculation shows that, for $d=\frac12$ (and only for this case),
there is another set of expressions having similar relations. That is, for
 \begin{align*}
  & \hat L_m\otimes\bb1 - \bb1\otimes \hat L_{-m} 
  + \lambda \sum_{k \in \bbR, \epsilon = \pm 1} Y_{\epsilon\alpha, k}\otimes  Y_{\epsilon\alpha,-m+k},
 \end{align*}
we calculate formally the commutators (this is only formal because for $d=\frac12$
we do not have the convergence for the product of two such expressions evaluated in a pair of vectors).
In Proposition \ref{pr:virasoro} we have seen that
\begin{align*}
	 \left[ \hat L_m\otimes \bb1 - \bb1\otimes \hat L_{-m}, 
	\sum_{k \in \bbR, \epsilon = \pm 1} Y_{\epsilon\alpha,k}\otimes  Y_{\epsilon\alpha, -n+k} \right] 
 	&=  \sum_{k \in \bbR, \epsilon = \pm 1} ((2d-1)m -n) Y_{\epsilon\alpha, k} \otimes  Y_{\epsilon\alpha, -m-n+k}
\end{align*}
and the full commutators are now
\begin{align*}
	& \left[\hat L_m\otimes \bb1 - \bb1\otimes \hat L_{-m} 
	+ \lambda \sum_{k \in \bbR, \epsilon = \pm 1} Y_{\epsilon\alpha,k}\otimes  Y_{\epsilon\alpha, -m+k},\right. \\
	&\qquad \left.\hat L_n\otimes \bb1 
	- \bb1\otimes \hat L_{-n} + \lambda \sum_{k \in \bbR, \epsilon 
	= \pm 1} Y_{\epsilon\alpha,k}\otimes  Y_{\epsilon\alpha, -n+k} \right]
	\\
	&=  (m-n)\hat L_{m+n}\otimes \bb1 + (-m-n)\otimes \hat L_{-m-n} 
	\\
	& \qquad + \lambda\sum_{k \in \bbR, \epsilon = \pm 1} (((2d-1)m - n) 
	- ((2d-1)n -m)) Y_{\epsilon\alpha, k} \otimes  Y_{\epsilon\alpha, -m-n+k} 
	\\
	&=  (m-n)\Bigl(\hat L_{m+n}\otimes \bb1 - \otimes \hat L_{-m-n} 
	+ 2d\lambda \sum_{k \in \bbR, \epsilon = \pm 1} Y_{\epsilon\alpha, k} \otimes  Y_{\epsilon\alpha, -m-n+k}\Bigr) \, .
\end{align*}
The last expression in the bracket coincides with
$\hat L_{m+n}\otimes \bb1 - \bb1\otimes \hat L_{-m-n} 
	+ \lambda \sum_{k \in \bbR, \epsilon = \pm 1} Y_{\epsilon\alpha,k}\otimes  Y_{\epsilon\alpha, -m-n+k}$
	if and only if $d=\frac12$.
The case $d=\frac12$ is related to free fermions.
This might indicate a hidden symmetry for free fermions.

\section{Outlook}
We need that the above generators are self-adjoint on a certain domain. This is open.
To show that they are essentially self-adjoint on our domain, one way would be to use the analytic vector
theorem, but it is unclear whether even $\Omega_0\otimes\Omega_0$ is an analytic vector.
Therefore, we need better estimates of the time-zero restriction
$Y_\alpha(w)\otimes Y_\alpha(w^{-1})$. Such estimates will be needed also to show
that the perturbed Lorentz generators do generate a new representation of the Lorentz group,
that we can construct a new Haag-Kastler net on $\dS^2$ and to find the interacting vacuum.
For this purpose, studying the Euclidean models of these two-dimensional CFT might help.

There are many two-dimensional CFTs and some of the charged primary fields have been relatively well-understood.
It might be a good idea to take other models where charged fields allow better control.

%

\subsection*{Acknowledgements}
Y.T.\! thanks Bin Gui for useful discussions on primary fields and estimates.
Y.T.\ is partially supported by
the \emph{MUR Excellence Department Project MatMod@TOV} awarded to the Department of Mathematics,
University of Rome Tor Vergata CUP E83C23000330006,
the University of Rome Tor Vergata funding \emph{OAQM} CUP E83C22001800005
and by GNAMPA--INdAM.

\appendix

\def\polhk#1{\setbox0=\hbox{#1}{\ooalign{\hidewidth
  \lower1.5ex\hbox{`}\hidewidth\crcr\unhbox0}}} \def\cprime{$'$}


\end{document}